\documentclass[11pt]{article}  

\usepackage{ifthen}
\usepackage{fullpage}
\usepackage{amsmath,amsfonts,amsthm}
\usepackage{enumitem}
\usepackage{graphicx}
\usepackage{algorithmic}
\usepackage{algorithm}

\newcommand{\F}{\mathcal{F}}
\newcommand{\dem}{\mathrm{D}}
\newcommand{\DT}{\tau}
\newcommand{\Paths}{\mathcal{P}}
\newcommand{\FP}{A}
\newcommand{\FT}{\tilde{A}}
\newcommand{\TP}{T^p}
\newcommand{\TF}{T^f}
\newcommand{\Pay}{\mathrm{Pay}}
\newcommand{\Cover}{\mathrm{Cov}}
\newcommand{\ST}{\mathcal{S}}
\newcommand{\cost}{\mathrm{cost}}
\newcommand{\D}{\mathcal{D}}
\newcommand{\X}{\mathcal{X}}
\newcommand{\Y}{\mathcal{Y}}
\newcommand{\U}{\Pi}
\newcommand{\OPT}{\mathrm{OPT}}
\newcommand{\meta}[1]{{\overline{#1}}}
\newcommand{\anc}{\mathrm{anc}}
\newcommand{\peer}{\mathrm{peer}}
\newcommand{\level}{\ell}
\newcommand{\FPay}{\mathrm{FPay}}

\newcommand{\argmin}{\operatorname{argmin}}
\newcommand{\argmax}{\operatorname{argmax}}

\newtheorem{theorem}{Theorem}

\newtheorem{lemma}[theorem]{Lemma}

\newtheorem{corollary}[theorem]{Corollary}

\title{Traffic-Redundancy Aware Network Design\thanks{This work was
    supported in part by NSF awards CCF-0643763 and CNS-0905134, and in part by a Sloan foundation fellowship.}}
\author{
Siddharth Barman\thanks{Computer Sciences Department, University of Wisconsin -- Madison. \tt{sid@cs.wisc.edu}.} 
\and
Shuchi Chawla\thanks{Computer Sciences Department, University of
  Wisconsin -- Madison. \tt{shuchi@cs.wisc.edu}.}  
}
\date{}

\begin{document}

\setcounter{page}{0}

\maketitle

\begin{abstract}
We consider network design problems for information networks where
routers can replicate data but cannot alter it. This functionality  allows the
network to eliminate data-redundancy in traffic, thereby saving on routing costs.
We consider two problems within this framework and design
approximation algorithms.

The first problem we study is the traffic-redundancy aware network
design (RAND) problem.  We are given a weighted graph over a single
server and many clients. The server owns a number of different data
packets and each client desires a subset of the packets; the client
demand sets form a laminar set system. Our goal is to connect every
client to the source via a single path, such that the collective cost
of the resulting network is minimized. Here the transportation cost
over an edge is its weight times times the number of {\em distinct}
packets that it carries.

The second problem is a facility location problem that we call
RAFL. Here the goal is to find an assignment from clients to
facilities such that the total cost of routing packets from the
facilities to clients (along unshared paths), plus the total cost of
``producing'' one copy of each desired packet at each facility is
minimized.

We present a constant factor approximation for the RAFL  and an $O(\log P)$ approximation for RAND, where $P$ is the total number of distinct packets. We remark that $P$ is always at most the number of different demand sets desired or the number of clients, and is generally much smaller.  
\end{abstract}

\thispagestyle{empty}
\newpage

\setcounter{page}{1}

\section{Introduction}
We consider network design problems for information
networks where edges can replicate data but cannot otherwise alter it.
In this setting our goal is to exploit the redundancy in the given
traffic matrix to save on routing costs. Formally, we are given a graph over a single server and many clients. The server has a universe of data packets available, and each client desires a subset of the packets. The goal is to determine a collection of paths, one from the source to each client, such that the total cost of routing is minimized. Here the cost of routing on an edge is proportional to the total size of the {\em distinct} packets that the edge carries. For example, if the edge belongs to two paths that each carry the same packet, then the edge only needs to route the packet once and not twice. We call this the traffic-redundancy aware network design problem, or RAND for short. 


 
RAND arises in networks that face a lot of data duplication. Consider for example a Netflix server serving movies to a large and varied clientele. Each client desires a certain subset of the movies. What routing paths should the server use to send the data to the clients so as to minimize its total bandwidth usage? If different movie streams involve disjoint sets of packets this boils down to setting up a single multicast tree, or solving the minimum Steiner tree problem, once per movie. However, the server may want to set up a single routing path for each client regardless of how many different movies the client desires. Moreover, clients desiring the same movie may desire it at different rates or qualities depending on their location or the device they are using. For example, a desktop user with a broadband connection may desire a high definition video, whereas a mobile phone user may be content with a much lower resolution. Then, the data sent to these clients is not identical but has some amount of overlap. The server can exploit this redundancy in traffic by using common paths for clients with similar demands, thereby saving on the actual amount of traffic routed. Redundancy in data can arise even across different movies when data streams are broken down into small enough packets. Anand et al.~\cite{AGA+08} show that this kind of traffic redundancy is highly prevalent in the Internet, and it can be eliminated across individual links by routers employing packet caches.




Given the cost structure that this redundancy generates, it makes sense to try to route the demands of clients desiring similar sets along overlapping paths. An extreme example of the benefit of merging paths is when all clients desire the same set of packets. In this case, the problem becomes equivalent to finding the minimum-cost Steiner tree over the clients and the server. At the other extreme, if all the clients desire disjoint sets of packets, then merging does not help at all, and it is optimal to pick the shortest path from every client to the source. There is thus a trade-off between routing demands along shortest paths and trying to merge the paths of clients with similar demands. 



We also study a facility location version of the problem that we call traffic-redundancy aware facility location or RAFL. This problem is motivated by the prevalence of content distribution networks (CDNs) in the Internet. Netflix servers, instead of connecting to clients directly, cache their data at multiple servers spread around the network that are hosted by a CDN such as Akamai; each client then connects to a CDN server individually to obtain its data. The savings in this case comes from assigning clients with similar demand sets to the same CDN server and sending a single copy of the multiply desired packets to the server. Formally we are given a network over potential facilities and clients. Each client, as before, desires a subset of the available packets. Our goal is to assign each client to a facility and route to each facility the union of the demand sets desired by clients assigned to it. The cost of such a solution is the sum of the cost of routing packets from facilities to clients (that are proportional to the size of the respective client's demand set), and the cost of routing packets to facilities (that is proportional to the total size of the distinct packets being routed to the facility). The savings from redundancy in this case are realized in the facility opening costs where we assign multiple clients with similar demand sets to the same facility and pay for each of the common packets only once.

RAND and RAFL model information networks as opposed to commodity networks in traditional network design problems. They can therefore be considered as intermediate models between traditional network design and network coding. In the latter the information network is allowed to use coding to increase its capacity. In our context the network can eliminate redundant information but cannot otherwise alter the information. See \cite{nwcoding} and references therein for work on network coding.  

\paragraph{Our results and techniques.}


We study the RAND and RAFL under a laminar demands assumption. In particular, we assume that the sets of packets demanded by clients form a laminar set family. In other words, every pair of demanded packet sets is either disjoint or one is a subset of the other. Such a structure arises, for example, in the Netflix problem described above when layered coding is used; the packets for a lower encoding rate are a subset of the packets for a higher encoding rate. 


RAND generalizes the minimum Steiner tree problem and RAFL generalizes metric uncapacitated facility location (MUFL); therefore both problems are NP-hard. We study approximations. 

We develop a constant factor approximation for the RAFL based on the natural LP-relaxation of the problem. Our algorithm follows the filtering approach developed by Lin and Vitter~\cite{LV92} and later exploited by Shmoys et al.~\cite{STA97} in the context of MUFL. In the MUFL setting, in the filtered solution, each client $t$ is associated with a set of facilities, say $\F(t)$, such that the routing cost of any of those facilities is comparable to the routing cost that the client pays in the LP solution. It is then sufficient to open a set of facilities in such a way that each can be charged to a client with a distinct set of associated facilities. Clients $t$ that are not charged (a.k.a. free riders) are rerouted to the facilities opened for other charged clients $t'$ such that $\F(t)$ and $\F(t')$ overlap, and furthermore the routing cost for $t'$ is smaller than that for $t$. Then, using the triangle inequality, the cost of rerouting can be bounded and this gives a constant factor approximation. In our setting, it is not sufficient to ensure that the routing cost of $t'$ is smaller than that of $t$. In addition, we must ensure that the facility opened by $t'$ can support the demand of $t$, otherwise every time we reroute a client we incur extra facility costs. Ensuring these two properties is tricky because clients with low routing costs may also have small demand sets; so essentially, these properties require us to consider clients according to two distinct and potentially conflicting orderings.

In order to deal with this issue, our algorithm is run in two phases. In the first phase we consider clients in order of increasing routing costs, in order to determine which clients will pay for their facilities and which ones are free riders. In the second phase, we consider free riders in decreasing order of the sizes of their demand sets. Each free rider is associated with a set of paying clients that pay for the facility that this client opens. Every time a free rider opens a facility, we reroute to it all of the other clients whose paying neighbors overlap with those of this facility. In this manner we can ensure that whenever a client is rerouted, it is routed to a facility that already produces the packets it needs. Unfortunately, our algorithm charges each paying client multiple times for different facilities opened. We ensure that the costs that a client pays each time it is charged form a geometrically decreasing sequence, and can therefore be bounded in terms of the LP solution. Overall we obtain a $27$-approximation.

For RAND an $O(\log n)$ randomized approximation can be obtained via tree embeddings~\cite{FRT03} because the problem is trivially solvable on trees; here $n$ is the number of nodes in the network. We give a simple deterministic combinatorial $O(\log P)$ approximation, where $P$ is the number of distinct packets to be routed. Note that if two or more packets are essentially identical in that they are desired by exactly the same set of clients, then we can combine them into a single packet (albeit with a larger size). Then, under the laminar demands assumption, $P$ is always at most the number of distinct demand sets or clients, which is at most $n$, the total number of nodes in the network. In fact in applications such as the Netflix multicast problem described above, we expect $n$ to be much larger than $P$.


Furthermore, our $O(\log P)$ approximation algorithm is a natural combinatorial algorithm that is simple and fast to implement. It is convenient to represent the laminar family of demand sets in the form of a tree where the demand set at any node is a proper subset of that at its parent and disjoint from those at its siblings. Our algorithm begins with some minor preprocessing of the demand tree to ensure that the tree has small (logarithmic) height. It then traverses the demand tree in a top-down fashion, and at each node of the tree constructs a (approximately optimal) Steiner tree over all the terminals with the corresponding demand set connecting them to the source. We show that the cost of the Steiner trees constructed at every level of the demand tree is bounded by a constant times the cost of the optimal solution, and therefore obtain an overall approximation factor proportional to the height of the tree. The height of the tree can however be much larger than $\log P$ because packets have different sizes. In order to prove an approximation factor of $O(\log P)$ we need to do a more careful bounding of the total cost spent on long chains of degree $2$ in the demand tree. Using the Steiner tree algorithm of Robins and Zelikovsky~\cite{RZ00} as a subroutine we obtain a $(6.2\log P)$-approximation.


\paragraph{Related work.}
RAND is closely related to single-source uniform buy-at-bulk network design (BaBND)~\cite{AA97, GKR03, MMP08}  in that both problems involve a trade-off between picking short paths between the source and the clients and trying to merge different paths to avail of volume discounts. However the actual cost structure of the two problems is very different. In BaBND the cost on an edge is a concave function of the total load on the edge; In our setting the cost is a submodular function of the clients using that edge. Neither of the problems is a special case of the other. BaBND admits a constant approximation in the single-source version~\cite{GKR03, Talwar02} and an $O(\log n)$ approximation is known for the general multi-source multi-sink problem~\cite{AA97}. 

One way of thinking about RAND is to break-up the problem and solution packet-wise: each packet $p$ defines a subset of the terminals, say $T_p$, that desire that packet; the solution restricted to these terminals is essentially solving a Steiner tree problem over the set $T_p\cup\{s\}$. Our goal is to pick a single collection of paths such that the sum over packets of the costs of these Steiner trees is minimized. In this respect, the problem is related to variants of the Steiner tree problem where the set of terminals is not precisely known before hand. This includes the maybecast problem of Karger and Minkoff~\cite{KM00} for which a constant factor approximation is known, as well as the universal Steiner tree problem~\cite{JLN+05} for which a randomized logarithmic approximation can again be obtained through tree embeddings and this is the best possible~\cite{BCK10}.

Facility location has been extensively studied under various models. The models most closely related to RAFL are the service installation costs model of Shmoys et al.~\cite{SSL04} and the heirarchical costs model of Svitkina and Tardos~\cite{ST06}. In the former, each client has a production cost associated with it; the cost of opening a facility is equal to a fixed cost associated with the facility plus the production costs of all the clients assigned to that facility. One way of representing these costs is in the form of a two-level tree for each facility with the fixed cost for the facility at the root of the tree and the client-specific production costs at the next level nodes. The cost of assigning a set of clients to the facility is the total cost of the subtree formed by the unique paths connecting the client nodes to the root of the tree. Svitkina and Tardos generalize this cost model to a tree of arbitrary depth, although in their setting the trees for different facilities are identical. Both the works present constant factor approximations for the respective versions, based on a primal-dual approach and local search respectively. Our model is similar to these models in that our facility costs are also submodular in the set of clients connecting to a facility. In our setting, the costs can again be modeled by a tree in which each node is associated with one or more clients; the cost of a collection of clients is given by the total cost of the union of subtrees rooted at those clients (as opposed to the portion of the tree ``above'' the clients). Therefore, neither of the two settings generalize the other. Moreover, facility costs in our setting are different for different facilities, although they are related through a multiplier per facility. Finally, while in Shmoys et al. and Svitkina et al. routing costs are given merely by the metric over facilities and clients, in our setting they are given by the distances times the total demand routed. 

For modeling information flow, Hayrapetyan et al.~\cite{hayrapetyan2005network} have studied a single-source network design problem with monotone submodular costs on edges, and a group facility location problem. The network design setting generalizes ours in that edge costs can be arbitrary submodular functions of the clients using the edges; they note that an $O(\log n)$ approximation can be achieved via tree embeddings. In contrast, we obtain an $O(\log P)$ approximation, where $P$, the number of distinct packets in the system, is always at most $n$ and generally much smaller. In the group facility location problem, edge costs are identical to those in RAND, but neither of the problems subsumes the other: the former assumes there are multiple facilities (sources) with fixed opening costs, but also limits the number of distinct packets per client to at most one.






\section{Problem definition and preliminaries}

The traffic-redundancy aware network design (RAND) problem is defined as follows. We are given a graph $G=(V,E)$ with weights $c_e\in \Re^+$ on edges $e\in E$, and a special node $s$ called the source. In addition, we are given a set $T$ of clients or terminals located at different nodes in the graph. The source carries a set $\U$ of packets with $|\U|=P$. Each packet $p\in \U$ is associated with a weight $w_p$; we assume that the weights are integral. Each terminal $t\in T$ desires some subset of the packets; this is called the terminal's demand and is denoted by $\dem(t)$. We use the convention $\dem(s)=\U$. Also let $w(S)=\sum_{p\in S}w_p$ denote the total weight of a set $S$ of packets.

We assume that the collection of demand sets $\D=\{\dem(t)\}_{t\in T}$ forms a laminar family of sets. In particular, for any two terminals $t_1, t_2\in T$, $\dem(t_1)\cap\dem(t_2)\ne \emptyset$ implies that either $\dem(t_1)\subseteq \dem(t_2)$ or $\dem(t_2)\subseteq\dem(t_1)$. We use a tree $\DT$ to represent the containment relationship between sets in the laminar family. The nodes of $\DT$ are sets in $\D$. A demand set $X$ is a parent of another set $Y$ if $Y\subset X$ and there is no set $Z\in\D$ with $Y\subset Z\subset X$. At the root of the tree is the universe $\U$ of packets. For a demand set $X$ in the tree $\DT$, we use $T_X$ to denote the terminals $t\in T$ with $\dem(t)=X$.


We denote an instance of RAND by the tuple $(G,T,\D,\DT)$.

The solution to RAND is a collection of paths $\Paths=\{P_t\}_{t\in T}$, with $P_t$ connecting the terminal $t$ to the source $s$. Given this solution, an edge $e\in E$ carries the set $S_{\Paths}(e) = \cup_{t\in T: P_t\ni e} \dem(t)$ of packets. The load on the edge is $w_{\Paths}(e) = w(S_{\Paths}(e))$. We drop the subscript $\Paths$ when it is clear from the context. The cost of the solution $\Paths$ is $\cost(\Paths)=\sum_{e\in E} c_e w_{\Paths}(e)$. Our goal is to find a solution of minimum cost.

Let $\OPT=\argmin_{\text{feasible }\Paths} \cost(\Paths)$ be an optimal solution. For a subset $W$ of terminals, we use $\OPT(W)$ to denote the restriction of $\OPT$ to $W$, that is, the collection of paths $\{P_t\in\OPT\}_{t\in W}$.


In the traffic-redundancy aware facility location problem (RAFL), we are given a set $\F$ of facilities, and a graph over $T\cup \F$ with edge weights $c_e$. Let $c(u,v)$ denote the shortest path distance between nodes $u$ and $v$ in the graph under the metric $c$. Furthermore, each facility $f\in \F$ has a cost $\lambda_f$ associated with it. 

A solution to this problem is an assignment $A$ from terminals to facilities. The assignment specifies the set of packets that a facility $f$ needs to produce in order to serve all the terminals connected to it: $S_{A}(f) = \cup_{t\in T: A(t)=f} \dem(t)$.  The load on the facility is $w_{A}(f) = w(S_{A}(f))$. We drop the subscript $A$ when it is clear from the context. The cost of the solution $A$ is $\cost(A)=\sum_{f\in\F} \lambda_f w_{A}(f) + \sum_{t\in T} w(\dem(t)) c(t,A(t))$. The first component of the cost is called the facility opening cost $C_f(A)$, and the second the routing cost $C_r(A)$ of the solution. Once again our goal is to find a solution of minimum cost.

Both RAND and RAFL are NP-hard because they generalize Steiner tree and metric uncapacitated facility location respectively. Our goal is to find approximation algorithms.

\section{A constant factor approximation for RAFL}

In this section we present a constant factor approximation for the RAFL. For ease of exposition we assume that all packets have unit weight, and write $|S|$ for the total size or weight of a set $S$ of packets. This assumption is without loss of generality. We further assume without loss of generality that $\min_{f\in\F} \lambda_f=1$.

The following is a natural LP-relaxation of RAFL. Here $x_{t,f}$ is an indicator for whether terminal $t$ is assigned to facility $f$, and $y_{f,p}$ denotes the extent to which facility $f$ produces packet $p$.

\begin{align*}
\textrm{minimize} & \ \   \sum_{f \in \F} \sum_{p \in \U} \lambda_f y_{f,p} + \sum_{t \in T} \sum_{f\in \F} |\dem(t)| x_{t,f} c(t,f) \\ 
\textrm{subject to  } & \ \ \sum_{f \in \F } x_{t,f} \geq 1 \ \ \  \forall t \in T \\ 
& y_{f,p} \geq x_{t,f} \ \ \  \forall t, f, p \in \dem(t) 
\end{align*}

Our approach begins along the lines of the filtering approach developed by Lin and Vitter~\cite{LV92} and employs some of the rounding ideas of Shmoys et al.~\cite{STA97} developed for metric uncapacitated facility location. In particular, given an optimal solution to the LP, we preprocess the solution at a constant factor loss in performance such that each terminal is assigned to a non-zero extent only to facilities for which the terminal's routing cost is within a small constant factor of the corresponding average amount in the LP. At this point, each terminal can be assigned to any facility to which it is fractionally assigned by the filtered solution, at a low routing cost. The key part of the analysis is bounding the cost for producing packets at facilities. 

Let $\F(t)$ denote the set of facilities to which $t$ is assigned fractionally by the filtered solution. In Shmoys et al.'s setting, in order to bound the cost of opening facilities, it is sufficient to find a ``paying'' terminal for each open facility such that the sets $\F(t)$ for paying terminals are mutually disjoint. For each terminal $t$ that is not paying, there exists at least one representative paying terminal $t'$ such that $\F(t)$ and $\F(t')$ overlap; any such terminal $t$ is assigned to the facility opened by its representative terminal $t'$ at a slight increase in routing cost as long as the average routing cost of $t'$ is no more than that of $t$. In order to accomplish this, we process terminals in order of increasing average routing cost.

In our setting, this approach has a basic flaw. The terminal $t'$ may have a much smaller demand set compared to the terminal $t$. Then, if we assign $t$ to the facility opened by $t'$, the facility needs to produce many more packets and its new larger opening cost can no longer be charged to $t'$. Unfortunately, it is not possible to ensure that $t'$ has both a small average routing cost than $t$ as well as a larger demand set. Instead, we divide the process of opening facilities into two parts. First we determine which terminals are paying and which ones are not by processing facilities in order of increasing average routing cost. Then we decide which facilities to open by processing the non-paying terminals in the order of decreasing demand set sizes.

The algorithm is described formally below. We first introduce some notation. Let $(x^*, y^*)$ denote the optimal solution to the RAFL LP given above. Let $C_r^*(t) = \sum_{f\in \F} x^*_{t,f}c(t,f)$ and $C_f^*(t)=\sum_{f\in \F} x^*_{t,f}\lambda_f$ denote the average routing and facility opening costs respectively under the solution $(x^*, y^*)$ associated with a terminal $t$. The total routing and facility opening costs of the solution are given by $C_r^*=\sum_{t\in T} |\dem(t)| C_r^*(t)$ and $C_f^*=\sum_{f\in \F} \sum_{p\in \U} y^*_{f,p}\lambda_f$ respectively. Likewise, for a feasible solution $(x,y)$, we use $C_r^{(x,y)}(t)$ and $C_f^{(x,y)}(t)$ to denote the average routing and facility opening costs associated with a terminal $t$ respectively. We drop the superscript $(x,y)$ when it is clear from context. Also let $C_r(x,y)$ and $C_f(x,y)$ denote the total routing and facility opening costs of the solution $(x,y)$.

\begin{algorithm*}{Given: LP solution $(x^*,y^*)$; Return: Assignment $\FP$ from terminals to facilities}
\caption{Rounding algorithm for RAFL}
\label{Algorithm:Round}

\textbf{Phase 1: Filtering}
\begin{algorithmic}[1]
\FORALL{$t\in T$}
\STATE Let $C_r^*(t) = \sum_{f\in \F} x^*_{t,f}c(t,f)$ and $C_f^*(t)=\sum_{f\in \F} x^*_{t,f}\lambda_f$
\STATE For all $f\in\F$, if $c(t,f)>\alpha C_r^*(t)$ set $x_{t,f}= 0$ else $x_{t,f}= x^*_{t,f}$.
\STATE \label{step:renormalize} Renormalize $x_{t,f}$ so that $\sum_f x_{t,f}=1$.
\STATE Let $\F(t) = \{f: x_{t,f}>0\}$, $C_f(t) =\sum_{f\in \F} x_{t,f}\lambda_f $, and $C_r(t) =  \sum_{f\in \F} x_{t,f}c(t,f) $.
\ENDFOR

\STATE For all $f\in\F$ and $p\in \U$, set $y_{f,p} = \max_{t\in T: \dem(t)\ni p} x_{t,f}$.

\vspace{11pt} 
\textbf{Phase 2: Classification of terminals into paying and free}
\vspace{2pt} 
\STATE  Initialize paying and free terminal sets: $\TP= \emptyset$ and $\TF = \emptyset$.  
\STATE For all $t \in T$ initialize temporary assignment $\FT(t) = \emptyset$, permanent assignment $\FP(t) = \emptyset$, and cover $\Cover(t) = \emptyset$.
\STATE For all $f \in \F$ initialize paying terminal set $\Pay(f) = \emptyset$ and final paying set $\FPay(f)=\emptyset$.

\WHILE{ $T \setminus \left( \TP \cup \TF \right) \neq \emptyset $ } 

    \STATE Let $t = \argmin_{j\in T\setminus(\TP\cup\TF)} C_r^*(j)$. \COMMENT{Select terminal with least connection cost}
   
    \IF{ there exists $f \in \F(t)$  such that  $ \Pay(f)$ covers $t$}
        \STATE $\TF =\TF \cup \{ t\}$  \COMMENT{$t$ is a free terminal}  
        \STATE $\FT(t)= f$
        \STATE Assign covering set for $t$:  $\Cover(t) = \{ j \in \Pay (f) \ | \ \dem(j) \cap \dem(t) \neq \emptyset  \}$ 
        \ELSE
        \STATE \label{step:add_to_TP} $\TP =\TP\cup\{ t\}$ \COMMENT{$t$ is a paying terminal}
        \STATE For all $f \in \F(t) $ update $\Pay(f) = \Pay(f)\cup\{ t\}$               
   \ENDIF 
    
    \ENDWHILE

\vspace{11pt} 
\textbf{Phase 3: Opening facilities}
\vspace{2pt} 
     
    \STATE For all $t \in \TP$ assign level $\level(t) = \lceil \log_2 C_f(t) \rceil$;
 for all $t \in \TF$ assign $\level(t) = \min_{ j \in \Cover(t) } \level(j) $.   
    \STATE Initialize $\TF_d =\{ t \in \TF \mid \level(t) = d \} $ and for all $t \in \TF_d$ set $\Gamma_d(t) = \{ t' \in \TF_d \mid \Cover(t) \cap \Cover(t') \neq \emptyset \} $.

    \FOR{$d=0$ to $\max_{ j \in \TP} \level(j)$}
    \STATE Initialize $W=\TF_d$  \COMMENT{Repeat till all terminals in $\TF_d$ have been assigned a permanent facility}.
    \WHILE{$ W \neq \emptyset$ } 
    \STATE Let $t \in \argmax_{ j \in W } |\dem(j)| $ \COMMENT{Select any terminal with the largest demand set}.
    \STATE Let $\bar{t} \in \argmin_{j \in \Gamma_d(t)} C_r^*(j)$.
    \STATE Let facility $\phi(\bar{t}) \in \argmin_{ f \in \cup_{ j \in \Cover(\bar{t})} \F(j) } \lambda_f $.
    \STATE \label{step:open_fac} Let $\FP(t)= \phi(\bar{t})$\COMMENT{We say that $t$ opens the facility $\phi(\bar{t})$.}
    \STATE \label{step:open_fac_free} For all $t'\in\Gamma_d(t) $, assign $\FP(t')=\FP(t)$. Update $W=W\setminus \left( \Gamma_d(t) + \{ t \} \right)$.
    
    \STATE Assign final paying set for facility $\FP(t)$:  $\FPay( \FP(t) ) = \Cover(t)$. 
   \ENDWHILE
   \ENDFOR
   
   \FORALL{ $t \in \TP$ }
   \STATE \label{step:open_fac_pay} Let $\FP(t) \in \argmin_{ f \in \F(t) } \lambda_f$.
   \ENDFOR
   
\end{algorithmic}
\end{algorithm*}

Our algorithm proceeds in three stages. The first is a filtering stage in which we convert the solution $(x^*, y^*)$ into a fractional solution $(x,y)$ which satisfies the following property: for all $t, f$ with $x_{t,f}>0$, $c(t,f)\le\alpha C_r^*(t)$. Here $\alpha$ is a parameter that we fix later. For a terminal $t$, we use $\F(t)$ to denote all the facilities $f$ that are fractionally assigned to $t$ in $(x,y)$, that is, have $x_{t,f}>0$.

In the second stage of the algorithm, we classify terminals into paying terminals $\TP$ and free terminals $\TF$. Essentially, a terminal $t$ becomes a free terminal if any of the facilities in $\F(t)$ is already expected to produce a large fraction of $t$'s demand. We record this facility as $t$'s temporary assignment $\FT(t)$. To this end, we say that a set of terminals $W$ covers a terminal $t$ if $\left| \dem(t) \setminus \left( \cup_{t' \in W } \dem(t') \right) \right| < \frac{1}{2}|\dem(t)|$. If a terminal $t$ is not covered at any of the facilities in $\F(t)$, then it becomes a paying terminal and can potentially pay for any of the facilities in $\F(t)$. $\Pay(f)$ tracks the set of terminals paying for a facility $f$.

Finally, in the third stage of the algorithm, we pick a permanent assignment from terminals to facilities by considering facilities in decreasing order of the sizes of their demand sets. As a first cut approach, suppose that we assign a free terminal $t$ to the facility at which it is covered ($\FT(t)$), and pay for that facility using the paying terminals associated with it. To ensure that no paying terminal $t'$ ends up paying for two or more opened facilities, we consider all the free terminals that this paying terminal covers and assign those also to the first facility that the paying terminal pays for. The order in which we assign free terminals to facilities ensures that in this last step we do not increase the facility opening cost of the solution. Here is the catch: which facility is actually opened is decided by the free terminal $t$ that starts this process and may be one of the more expensive facilities in the paying terminal $t'$'s set $\F(t')$. In this case, the paying terminal does not have enough charge in the LP solution to pay for this facility. In order to avoid this situation, we consider all of the facilities that are ``close'' to $t$ or to other free terminals covered by the paying terminals that cover $t$. Of these we open the facility with minimum cost and pay for it using the paying terminals associated with $f$. 

While in our algorithm a paying terminal can end up paying for multiple opened facilities, in Lemma~\ref{lemma:fac_cost_one} below we argue that the costs of those facilities decrease geometrically and so the sum can be bounded.

We now formalize this argument. We begin by showing that the fractional solution $(x,y)$ is not too expensive.

\begin{lemma}
\label{lemma:filter_opt}
The solution $(x,y)$ is feasible for the RAFL LP. Moreover 
$C_f(x,y) \leq \frac{\alpha}{\alpha -1 } C_f^*$.
\end{lemma}
\begin{proof}
  $(x,y)$ is feasible by construction. Note that for all $t\in T$, $\sum_f x^*_{f,t}=1$, otherwise the cost of the solution can be improved. Therefore, by Markov's inequality, $\sum_{f: c(f,t)>\alpha C_r^*} x^*_{f,t} \le 1/\alpha$. Then, in the renormalization step (Step~\ref{step:renormalize}) we set $x_{t,f}$ to be no more than $\alpha/(\alpha-1) x^*_{t,f}$. 
   
For all $f \in F$ and $p \in \U$ we set $y_{f,p}$ to be  $\max_{t\in T: \dem(t)\ni p} x_{t,f}$. Hence $y_{f,p}$ is no more than $ \max_{t\in T: \dem(t)\ni p} \frac{\alpha} {\alpha - 1} x^*_{t,f}$, which is no more than $\alpha/(\alpha-1) y^*_{f,p}$. This in turn implies the second part of the claim. 
\end{proof}

Next we bound the routing cost of the assignment $\FP$. 

\begin{lemma}
\label{lemma:routing}
For all $ i \in T$, $c(i,\FP(i))\le 9 \alpha C_r^*(i)$. 
\end{lemma}
\begin{proof}
  Let $i$ be a paying terminal that is assigned a facility $f$ in Step~\ref{step:open_fac_pay} of the algorithm. Then, $f\in\F(i)$ and therefore, by the definition of $x$ and $\F(i)$, $c(i,f)\le \alpha C_r^*(i)$. 
  
Now, let $i$ be a free terminal which is assigned a facility $f$ in Step~\ref{step:open_fac} or Step~\ref{step:open_fac_free} of the algorithm. Say level $\level(i) = d$ and let $t$ be the terminal that opened $f$. Then there are two possibilities: either $i=t$ or $i\in\Gamma_d(t)$. In the former case, $\FT(t) \in \F(i)$ and so $c(i, \FT(t)) \leq \alpha C_r^*(i)$.

Next we show that if terminal $t$ opens a facility (in Step \ref{step:open_fac}) then for all $t' \in \Gamma_d(t)$ we have $c(t', \FT(t)) \leq 3 \alpha C_r^*(t')$. Note that $\Cover(t') \cap \Cover(t) \neq \emptyset$ and let $j$ be a paying terminal in $\Cover(t') \cap \Cover(t)$. We have $\FT(t) \in \F(j)$ and  $\FT(t') \in \F(j)$. Since terminal $j$ was selected in the second phase before $t'$ we have $C^*_r(j) \leq C^*_r(t')$. By triangle inequality we get that $c(t', \FT(t)) \leq c(t', \FT(t') ) + c (\FT(t'), j) + c(j , \FT(t) ) \leq \alpha C_r^*(t') + 2 \alpha C_r^*(j) \leq 3 \alpha C_r^*(t')$.

In Step~\ref{step:open_fac} the opened facility, say $f$, is selected to be $\phi( \bar{t})$ for the terminal $\bar{t} \in \Gamma_d(t)$ with minimum $C_r^*$ value. Since $ f \in \F(j')$ for some $j' \in \Cover(\bar{t})$ we have $ c(f, \bar{t}) \leq 3 \alpha C_r^*(\bar{t}) $. Also $\bar{t}$ is contained in $\Gamma_d(t)$, which implies that $c(\bar{t}, \FT(t) ) \leq 3 \alpha C_r^*(\bar{t})$.  

Again using triangle inequality we get the desired result: $c(i,f) \leq c(i , \FT(t) ) + c(\FT(t), \bar{t}) + c(\bar{t}, f) \leq 3 \alpha C_r^*(i) + 6 \alpha C_r^*(\bar{t}) \leq 9 \alpha C_r^*(i)$.  Here the last inequality follows from the definition of $\bar{t}$. 
\end{proof}

Finally we account for the facility opening cost of the solution. Recall that $S_{\FP}(f)$ denotes the set of packets produced at $f$ under the assignment $\FP$. We note that a facility $f$ may be ``opened'' multiple times by different free terminals in Step~\ref{step:open_fac} of the algorithm. In this case, we treat each subsequent opening as opening a new copy of $f$ and designate a distinct set of terminals, $\FPay(f)$, to pay for all of the packets to be produced at the new copy freshly (even though some of them may already be assigned to the facility). 

Henceforth, for ease of exposition we will assume that each facility is opened at most once.

The following lemma notes that $|S_{\FP}(f)|$ can be bounded in terms of the demand sets of the terminals finally paying for this facility.
\begin{lemma}
\label{lemma:saf}
  Let $f$ be a facility opened by a free terminal $t$. Then $S_{\FP}(f) = \dem(t)\cup \cup_{j\in\Cover(t)} \dem(j)$. Furthermore, $|S_{\FP}(f)|\le 2|\cup_{j\in\Cover(t)} \dem(j)|$.
\end{lemma}
\begin{proof}
  Let $t\in\TF_d$ and consider a terminal $t' \in \Gamma_d(t)$. We claim that $\dem(t')\subseteq \dem(t)\cup \cup_{j\in\Cover(t)} \dem(j)$, and this implies the first part of the lemma. Let $j\in\Cover(t)\cap\Cover(t')$. Then, by the definition of a covering set, $\dem(t')\cap\dem(j)\ne\emptyset$. By laminarity, either $\dem(j)\supset\dem(t')$ in which case our claim holds, or $\dem(j)\subset\dem(t')$. In the latter case, $\dem(t)\cap\dem(j)\ne\emptyset$ implies $\dem(t)\cap\dem(t')\ne\emptyset$. Once again by laminarity, either $\dem(t')\subseteq\dem(t)$, or $\dem(t')\supset\dem(t)$. The latter case cannot hold because $t$ is considered before $t'$ in phase $3$ and therefore $|\dem(t')|\le|\dem(t)|$. In the former case our claim holds.

The second part of the lemma follows from the definition of covering.
\end{proof}

Using this lemma we can bound the facility opening cost of the assignment $\FP$ in terms of the average facility opening costs $C_f(t)$ for the paying terminals $t\in \TP$ in the fractional solution $(x,y)$. 
\begin{lemma}
\label{lemma:fac_cost_one}
  $\sum_f |S_{\FP}(f)|\lambda_f \le 9 \sum_{t\in\TP} |\dem(t)| C_f(t)$.
\end{lemma}
\begin{proof}
First we bound the opening cost of facilities that were opened by free terminals. Lemma~\ref{lemma:saf} implies that $|S_{\FP}(f)|\le 2|\cup_{j\in\FPay(f)} \dem(j)|\le 2\sum_{j\in\FPay(f)}|\dem(j)|$. For facility $f$ write $\level(f) =d$ iff the terminal opening it has level $d$. Note that for facility $f$ with $\level(f) = d$ we have $\lambda_f \leq 2^d$. 
  
Fix a terminal $i \in \TP$, and write $\Cover^{-1}(i) = \{  j \in \TF \mid i \in \Cover(j) \}$. Say $\level(i) = \ell$. Then $C_f(i) \in (2^{\ell-1}, 2^{\ell} ]$. By definition, for all $j \in \Cover^{-1}(i)$ we have $\level(j) \leq \ell$. 
 
We claim that $i$ pays for at most one facility at level $d$ for $d\le \ell$, and does not pay for any facilities at level $d>\ell$.  We prove the first part by contradiction. Say there exits $f \neq f'$ such that $\level(f) = \level(f')=d$ and $i \in \FPay(f) \cap \FPay(f')$. Write $t$ as the terminal that opens $f$ and $t'$ as the terminal that opens $f'$. Both $t$ and $t'$ are in $\Cover^{-1}(i)$ and have level equal to $d$. Without loss of generality assume $t$ was processed before $t'$ by the algorithm. Then $t' \in \Gamma_d(t)$ and we would have $\FP(t') = \FP(t)$, contradicting the assumption that they are assigned to different facilities. 

For the second part of the claim, suppose that $i \in \FPay(f)$ for some facility $f$. Then $f$ is opened by a terminal $t \in \Cover^{-1}(i)$. As stated above, the level of such a terminal $t$ must be no more than the level of $i$, hence $\level(f) \leq \level(i)$.

For a level $d$ facility $f$ we have $\lambda_f \leq 2^d$ and hence $\sum_{f  : \FPay(f) \ni i } \lambda_f \leq \sum_{d=1}^\ell 2^d \leq 4 C_f(i) $. We therefore get the following chain of inequalities.
\begin{align*}
\sum_f \lambda_f |S_{\FP}(f)| & \le 2 \sum_f \lambda_f \sum_{j\in\FPay(f)}  |\dem(j)| \\
& = 2 \sum_{ j \in \TP } | \dem(j) | \sum_{f : \ \FPay(f) \ni j } \lambda_f  \\
& \le  8  \sum_{j\in \TP} C_f(j) \ |\dem(j)|
\end{align*}
Here the first inequality follows from Lemma \ref{lemma:saf} and the second inequality follows from the bound $\sum_{f  : \FPay(f) \ni i } \lambda_f \leq 4 C_f(i)$. 

To account for facilities that were opened by paying terminals in Step \ref{step:open_fac_pay} of the algorithm we note that for any such facility $h$, the set $S_{\FP} (h) = \dem(t)$ where $t$ is the paying terminal that opened $h$. Moreover $\lambda_{h} \leq \sum_f x_{t,f} \lambda_f $ and hence the opening cost incurred by the algorithm is no more than the fractional value, $\lambda_h | \dem(t) | \leq C_f(t) \ |\dem(t)|$. Hence the total facility opening cost incurred by the algorithm is no more than $9 \sum_{t\in\TP} |\dem(t)| C_f(t)$. 
\end{proof}

To complete the argument, we relate the costs $C_f(t)$ to the total cost $C_f(x,y)$.
\begin{lemma}
\label{lemma:fac_cost_two}
  $\sum_{t\in\TP} |\dem(t)| C_f(t)\le 2C_f(x,y)$.
\end{lemma}
\begin{proof}
  When a terminal $t$ is added to $\TP$ (Step~\ref{step:add_to_TP} of the algorithm) it is not covered at any of the facilities in $\F(t)$. For $f\in\F$ let $L(t,f) = \dem(t)\setminus\cup_{t'\in\Pay(f)}\dem(t')$ denote the set of packets in $\dem(t)$ that is not covered at $f$ at the time that $t$ is considered. Here, $\Pay(f)$ denotes the set of terminals that is paying for $f$ at the time that $t$ is considered. Note that $|L(t,f)|\ge 1/2 |\dem(t)|$ for $f\in\F(t)$. For any facility $f$ the sets $L(t,f)$ are disjoint and partition the support of $y_{f,p}$ and hence the following chain of inequalities hold:

\begin{eqnarray}
\sum_{p} y_{f,p} & \geq & \sum_{t \in \TP}   \sum_{p \in L(t,f)} y_{f,p}   \nonumber \\ 
 & \geq & \sum_{t \in \TP}  \sum_{p \in L(t,f)} x_{t,f}  \nonumber \\
 & = & \sum_{t \in \TP} |L(t,f)| x_{t,f}  \label{eqn:Ltf} 
\end{eqnarray}

We can now derive the desired bound:
\begin{align*}
C_f(x,y) & =\sum_f \sum_p \lambda_f y_{f,p} \\
& =  \sum_f \lambda_f \sum_p y_{f,p} \\
& \geq  \sum_f \lambda_f \sum_{t \in \TP} |L(t,f)| \ x_{t,f} \\
& \geq  \sum_f \sum_{ t \in \TP} \frac{1}{2} \ |\dem(t)| \ \lambda_f x_{t,f} \\
& =  \sum_{t \in \TP} \frac{1}{2} \ |\dem(t)| \sum_f \lambda_f x_{t,f} \\
& =  \frac{1}{2} \sum_{t \in \TP}  \ |\dem(t)| C_f^{(x,y)}(t)
\end{align*}
Here the third step follows from inequality (\ref{eqn:Ltf}) and the fourth step holds because for any $t \in \TP$ and $f\in\F$, either $x_{t,f}=0$ or $f\in\F(t)$ and the size of the set $L(t,f)$ is at least half its demand: $|L(t,f)| \geq 1/2 |\dem(t)|$.
\end{proof}

Putting together the above lemmas we obtain the following theorem:
\begin{theorem}
  Algorithm~\ref{Algorithm:Round} gives a $27$-approximation to the RAFL.
\end{theorem}
\begin{proof}
Lemma \ref{lemma:routing} implies that for any terminal the routing cost is no more than $9 \alpha$ times the optimal. Also from Lemma \ref{lemma:fac_cost_one} and Lemma \ref{lemma:fac_cost_two} we get that the total facility opening cost is no more than $18$ times the total facility opening cost of the filtered solution $C_f(x,y)$. Finally from Lemma \ref{lemma:filter_opt} we have $C_f(x,y) \leq \alpha/( \alpha -1)  \ C_f^*$, so the facility opening cost of the generated solution is no more than $18 \alpha / ( \alpha -1) $ times the optimal. Hence the algorithm achieves an approximation factor of $\max \{ 9 \alpha , \frac{18 \alpha}{\alpha -1 } \}$, which is minimized at $\alpha=3$ to give us a $27$-approximation.   
\end{proof}


\section{An $O(\log P)$ approximation for RAND}
In this section we develop an $O( \log P)$-approximation algorithm for the RAND where $P=|\U|$. The basic observation that our algorithm hinges on is that if every pair of demand sets in $\D$ is either identical or disjoint, that is, the tree $\DT$ is a two-level tree, then the problem becomes easy and can be approximated to within a small constant factor. In particular, then the problem becomes one of finding optimal Steiner trees connecting the terminals in $T_X$ to $s$ for every set $X\in \D$. The cost of these Steiner trees is purely additive because the different sets in $\D$ are disjoint.

In particular, this implies that for any collection $\X$ containing disjoint sets of packets, we can construct a partial solution over terminals in $\cup_{X\in\X} T_X$ at a cost of constant times the cost of the optimal solution. This immediately suggests an algorithm with approximation ratio a constant times the depth of the tree $\DT$: define the level of a node in $\DT$ as its distance from the root $\U$; for every level $k$, consider the collection $\X_k$ of sets at level $k$ and construct a constant factor approximation over terminals in $\cup_{X\in\X_k} T_X$. 

In order to obtain a small approximation ratio through this approach, our algorithm first performs a preprocessing of the demand tree $\DT$ to ensure, at small cost, that for any pair of demand sets in the tree where one is a parent of the other, the total weight of the parent is at least twice as large as the total weight of the child. Since the weight of every packet is integral, this implies that the depth of the preprocessed tree is at most $\log w(\U)$, and gives us an $O(\log w(\U))$ approximation.

To obtain an $O(\log P)$ approximation, we need to do a more clever analysis. As discussed in the introduction, $P$ denotes the number of effectively distinct packets in the instance. In particular, we can assume without loss of generality that $P$ is equal to the number of nodes in the tree $\DT$. Our next key observation is that we can collectively bound the total cost of Steiner trees for nodes in a long root to leaf chain of nodes by a constant times the cost of the optimal solution (rather than by the length of the chain times the optimal cost). Here we crucially use the fact that each subsequent node in the chain has a weight at most half that of the preceding node, a fact that is ensured through our preprocessing step. Given this, we break up the demand tree into $\log P$ collections of chains, each of which corresponds to disjoint packet sets. The total cost of Steiner trees over each collection of chains can then be bounded to within a constant factor of the optimal cost, and we get an overall $O(\log P)$ approximation.


We now present our algorithm and analysis formally.

\paragraph{Preprocessing the tree $\DT$.}
Our preprocessing phase is described in Algorithm~\ref{Algorithm:Preproc} below. The preprocessing phase makes two changes to the given instance. First, it changes the demand sets of some terminals to supersets of their original demands. Second, after these changes to demand sets, if there are nodes in $\DT$ that do not have any terminals associated with them, it merges these nodes with their parents.

\begin{algorithm}{Given: Instance $(G,T,\D,\DT)$; Return: New instance $(G,T,\D',\DT')$}
\caption{Preprocessing algorithm}
\label{Algorithm:Preproc}
\begin{algorithmic}[1]
\STATE Perform depth first search over the tree $\DT$.
\FORALL{ nodes $X\in\D$ encountered during DFS}
\STATE Let $Y$ be the parent of $X$ in $\DT$
\IF{$w(X)>\frac12 w(Y)$}
\STATE For all $t\in T_X$, set $\dem'(t)=Y$.
\STATE Merge $X$ with $Y$. That is, set $T_Y=T_Y\cup T_X$, remove $X$ from $\DT$, and reattach the children of $X$ in $\DT$ as the children of $Y$ in the modified tree.
\ENDIF
\ENDFOR
\end{algorithmic}
\end{algorithm} 

We obtain the following lemma.
\begin{lemma}
\label{lem:preproc}
  Consider an instance $(G,T,\D,\DT)$ of RAND. Then Algorithm~\ref{Algorithm:Preproc} returns an instance $(G,T,\D',\DT')$ with the following properties:
  \begin{enumerate}
  \item For every $t\in T$, $\dem'(t)\supseteq \dem(t)$ and $w(\dem'(t))\le 2w(\dem(t))$.
  \item $\D'$ is a laminar family.
  \item For every pair of sets $X,Y\in\D'$ such that $X$ is a child of $Y$ in $\DT'$, $w(Y)\ge 2w(X)$.
  \item The cost of the optimal solution over $(G,T,\D',\DT')$ is at most twice that of the optimal solution over $(G,T,\D,\DT)$.
  \item Any feasible solution to $(G,T,\D',\DT')$ is also feasible for $(G,T,\D,\DT)$.
  \end{enumerate}
\end{lemma}
\begin{proof}
  We first remark that the demand of any terminal $t$ is modified at most once, when its original node $\dem(t)$ is encountered in the DFS over $\DT$; after this, the new node $\dem'(t)$ is never processed again by DFS. Then, it is easy to see that the first property holds. The second property holds because $\DT'$ is a tree over sets in $\D'$ and therefore any two sets in $\D'$ are either disjoint or one is an ancestor of another. The third property holds by construction: when the node $X$ is encountered by the DFS, if the node survives the preprocessing, then it holds that $w(X)\le \frac12w(Y)$.
The fifth property follows immediately from the first.

  To prove the fourth property, we note that for any two terminals $t_1$ and $t_2$, $\dem(t_1)\subset \dem(t_2)$ implies $\dem'(t_1)\subseteq\dem'(t_2)$. Now consider any optimal solution $\Paths=\{P_t\}_{t\in T}$ to $(G,T,\D,\DT)$. We will show that the cost of the solution $\Paths$ for the new instance $(G,T,\D',\DT')$ is no more than twice its cost for $(G,T,\D,\DT)$. For every edge $e\in E$ consider the set of terminals, say $T(e)$, that use $e$: $T(e)=\{t: P_t\ni e\}$. Let $T'(e)$ be the subset of $T(e)$ of terminals whose demand is not contained inside the demand of any other terminals in $T(e)$; that is $T'(e)$ denotes the terminals with the maximal demand sets. Then, in the instance $(G,T,\D,\DT)$, $w(S_{\Paths}(e)) = \sum_{t\in T'(e)} w(\dem(t))$. Moreover by our observation above, in the new instance $(G,T,\D',\DT')$, terminals in $T'(e)$ are still the maximal demand terminals and $w(S_{\Paths}(e)) = \sum_{t\in T'(e)} w(\dem'(t))$. The claim now follows from the first property.
\end{proof}



\paragraph{Main algorithm.}

We now proceed to describe the main algorithm. Henceforth we will assume that the given instance of RAND satisfies the properties listed in Lemma~\ref{lem:preproc}, in particular, property 3. 

\begin{algorithm}{Given: Instance $(G,T,\D,\DT)$ satisfying the properties in Lemma~\ref{lem:preproc}; Return: Collection of paths from each terminal to the source $s$.}
\caption{A logarithmic approximation for RAND}
\label{Algorithm:log-p}
\begin{algorithmic}[1]
\STATE Perform depth first search over the tree $\DT$.
\FORALL{ nodes $X\in\D$ encountered during DFS}
\STATE Construct an approximately optimal Steiner tree $\ST(X)$ in $G$ over $T_X\cup\{s\}$. \label{step:steiner-tree}
\STATE For every $t\in T_X$, return the unique path in $\ST(X)$ from $t$ to $s$.
\ENDFOR
\end{algorithmic}
\end{algorithm} 
 
We first define some notation. For a demand set $Y$, let $\ST^*(Y)$ denote the optimal Steiner tree over $T_Y\cup\{s\}$. The cost of this Steiner tree is $\cost(\ST^*(Y)) = w(Y)\sum_{e\in\ST^*(Y)} c_e$. Analogously we define the cost of an arbitrary Steiner tree $\ST(Y)$ over $T_Y\cup\{s\}$ as $\cost(\ST(Y))$.

In our analysis, we sometimes need to consider sets of nodes in $\DT$ and bound their cost collectively. To this end, we define a {\em chain} $\meta{Y}$ to be a set $\meta{Y} = \{Y_0, Y_1, \cdots, Y_k\}$ where for every $i<k$, $Y_i$ is a parent of $Y_{i+1}$. Recall that this implies $Y_0\supset Y_1\supset\cdots\supset Y_k$ and $w(Y_i)\ge 2w(Y_{i+1})$ for all $i<k$. We call the node $Y_0$ the start of the chain $\meta{Y}$. We say that two chains $\meta{Y_1}$ and $\meta{Y_2}$ are disjoint if $(\cup_{Y\in \meta{Y_1}}  Y)\cap(\cup_{Y\in \meta{Y_2}}  Y)=\emptyset$.




The following is the main lemma of this section and allows us to bound the cost of large collections of nodes in $\DT$.

\begin{lemma}
\label{lem:disjoint-sets}
  Let $\Y=\{\meta{Y}_1, \meta{Y}_2, \cdots\}$ be a collection of mutually disjoint chains. That is, for any $\meta{Y}_i, \meta{Y}_j\in\Y$, $\meta{Y}_i$ and $\meta{Y}_j$ are disjoint. Then $\sum_{\meta{Y}\in\Y}\sum_{Y\in\meta{Y}} \cost(\ST^*(Y))$ is at most $2\cost(\OPT)$.
\end{lemma}
\begin{proof}
For $\meta{Y}\in\Y$, let $\OPT(\meta{Y})$ denote $\OPT(\cup_{Y\in\meta{Y}} T_Y)$, the restriction of the optimal solution to the terminals associated with sets in $\meta{Y}$. Note that because the chains in $\Y$ are disjoint, $\sum_{\meta{Y}\in\Y} \cost(\OPT(\meta{Y}))\le \cost(\OPT)$. Therefore, for the rest of the proof we will argue that for any $\meta{Y}\in\Y$, $\sum_{Y\in\meta{Y}} \cost(\ST^*(Y))\le 2\cost(\OPT(\meta{Y}))$, and this will imply the lemma.

To prove the claim, fix a chain $\meta{Y}\in\Y$, and let $\meta{Y} = \{Y_0, Y_1, \cdots, Y_k\}$ where $Y_0\supset Y_1\supset\cdots\supset Y_k$, and $w(Y_i)\ge 2w(Y_{i+1})$ for all $i<k$. Recall that $\OPT(\meta{Y})$ contains a path for every terminal in $\cup_{Y\in\meta{Y}} T_Y$. Let $\Paths_Y$ denote the collection of paths for terminals in $T_Y$. We say that $e\in\Paths_Y$ if $e\in \cup_{t\in T_Y} P_t$. Let $\ST(Y)$ denote the Steiner tree over $T_Y\cup\{s\}$ defined by $\Paths_Y$ and note that $\cost(\ST^*(Y)) \le \cost(\ST(Y)) = w(Y)\sum_{e\in\Paths_Y} c_e$. Then, 
\[\sum_{Y\in\meta{Y}} \cost(\ST^*(Y))\le \sum_{e\in E} c_e \sum_{Y\in\meta{Y}: \Paths_Y\ni e} w(Y)\] 
On the other hand, because of the containment structure of sets in $\meta{Y}$, \[\cost(\OPT(\meta{Y})) = \sum_{e\in E} c_e \max_{Y\in\meta{Y}: \Paths_Y\ni e} w(Y)\]

To conclude the proof we claim that for every edge $e$, $\sum_{Y\in\meta{Y}: \Paths_Y\ni e} w(Y)\le 2\max_{Y\in\meta{Y}: \Paths_Y\ni e} w(Y)$. But this is easy to see because weights of sets $Y\in\meta{Y}$ are geometrically decreasing by a factor of at least two.
\end{proof}

Next we show that the tree $\DT$ can be decomposed into at most $\log P$ different collections of mutually disjoint chains. This along with Lemma~\ref{lem:disjoint-sets} will allow us to prove our desired approximation.
\begin{lemma}
\label{lem:decomp}
  Any demand tree $\DT$ can be decomposed into at most $\log P$ different collections of mutually disjoint chains $\Y_1, \Y_2, \cdots, \Y_k$, such that each node in $\DT$ belongs to exactly one collection.  Here $P$ is the number of nodes in $\DT$.
\end{lemma}
\begin{proof}
  We decompose the tree by finding a long chain, removing it from the tree, and then recursing on the remaining subtrees. Given a tree $\DT$ with root $Y_0$, we start at $Y_0$ and follow a path down to a leaf. Let $m$ denote the number of nodes in $\DT$.  At any node, we consider the sizes of the subtrees rooted at the node in terms of the number of nodes in the tree; the path then moves to the child corresponding to the largest subtree. Note that all of the other remaining subtrees rooted at the node have at most $m/2$ nodes. 

Let $\meta{Y}=\{Y_0, Y_1, \cdots, Y_k\}$ be the path obtained. This collection of nodes is a chain by definition. Consider removing the nodes in $\meta{Y}$ from $\DT$. We then note the following properties: (1) Every remaining connected component of the tree is of size at most $m/2$; (2) Let $\DT_1$ and $\DT_2$ denote any two connected components. Then $(\cup_{X\in\DT_1} X)\cap(\cup_{X\in\DT_2} X)=\emptyset$, that is, the connected components are mutually disjoint.

The second property can be proved by contradiction. If two of the components are not disjoint, then by laminarity, they contain nodes $X_1$ and $X_2$ respectively such that $X_1$ is an ancestor of $X_2$ in $\DT$. Since $X_1$ and $X_2$ belong to different connected components, there must be a node on the path between them in $\DT$ that is in the chain $\meta{Y}$. Then, all ancestors of this node are in $\meta{Y}$ including $X_1$ which contradicts the fact that $X_1$ belongs to a connected component left after removing the chain.

We output $\{\meta{Y}\}$ as the first collection of mutually disjoint chains. (Note that this first collection has only one chain in it.) Now let $\DT_1, \DT_2, \cdots$ be the connected components left over in $\DT$ after removing the nodes in $\meta{Y}$. We recursively find collections of mutually disjoint chains in each of the components. Call these $\Y_{i1}, \Y_{i2}, \cdots$ etc. for the $i$th connected component. Then, we output the collections $\cup_i \Y_{ij}$ for each $j$. Since the connected components are mutually disjoint, the collections we output are also mutually disjoint. Furthermore, since the sizes of the connected components decrease by a factor of $2$ each time we go down a level of recursion, it is easy to argue that the number of collections we output are bounded by $\log P$ where $P$ is the number of nodes in the tree $\DT$ that we started out with.
\end{proof}

We now present the main theorem for this section:
\begin{theorem}
\label{thm:main-log-p}
  Let $(G,T,\D,\DT)$ be an instance of RAND that satisfies the conditions in Lemma~\ref{lem:preproc}. Then Algorithm~\ref{Algorithm:log-p} obtains a $(2\alpha\log P)$-approximation for the RAND over this instance where $P$ is the number of effectively distinct packets in the instance, and $\alpha$ is the approximation factor of the Steiner tree algorithm used in Step~\ref{step:steiner-tree} of the algorithm.
\end{theorem}
\begin{proof}
  For every set $X\in\D$, let $\ST(X)$ denote the Steiner tree built by our algorithm over $T_X\cup\{s\}$. Then we have $\cost(\ST(X))\le\alpha\cost(\ST^*(X))$. We first use Lemma~\ref{lem:decomp} to decompose the tree $\DT$ into at most $\log P$ collections of mutually disjoint chains. Call these collections $\Y_1, \Y_2, \cdots, \Y_k$.


The total cost of our solution can now be written as
\begin{align*}
\sum_{Y\in\D} \cost(\ST(Y)) & = \sum_{i\le k} \sum_{\meta{Y}\in\Y_i}\sum_{Y\in\meta{Y}} \cost(\ST(Y))\\
& \le \alpha \sum_{i\le k} \sum_{\meta{Y}\in\Y_i}\sum_{Y\in\meta{Y}} \cost(\ST^*(Y))\\
& \le \alpha \sum_{i\le k} 2\cost(\OPT)\\
& = 2\alpha k\ \cost(\OPT)
\end{align*}
Here the second inequality follows by applying Lemma~\ref{lem:disjoint-sets}. The theorem now follows by noting that $k\le\log P$.
\end{proof}

Combining Theorem~\ref{thm:main-log-p} with Lemma~\ref{lem:preproc} we get the following result.
\begin{corollary}
  Algorithms~\ref{Algorithm:Preproc} and \ref{Algorithm:log-p} together obtain a $(4\alpha\log P)$-approximation for the RAND where $\alpha$ is the approximation factor of the Steiner tree algorithm used in the algorithm.
\end{corollary}

We conclude this section by noting that Algorithm~\ref{Algorithm:log-p} can be implemented in a simple combinatorial fashion in $O(n^3)$ time as a generalization of Prim's algorithm for the minimum spanning tree problem as follows. This version of the algorithm obtains an $(8\log P)$-approximation.
\begin{enumerate}
\item Let $\anc(t)$ denote the set of ancestors of $t$ (those with demands that are strict supersets of $\dem(t)$), and $\peer(t)$ denote the set of its peers (those with demands identical to that of $t$).
\item At any step call a terminal $t$ eligible if all of its ancestors are already connected to the root.
\item Initialize $W=\{s\}$.
\item Let $\Delta(t)$ denote the distance in $G$ from $t$ to its closest node in $W\cap(\anc(t)\cup\peer(t))$.
\item While $W\ne T$, pick the eligible terminal with the smallest $\Delta(t)$ and connect it to its closest node in $W\cap(\anc(t)\cup\peer(t))$. Update $W=W\cup\{t\}$ and update $\Delta$.
\end{enumerate}

\bibliographystyle{plain}
\bibliography{rand}

\end{document}